\documentclass{llncs}
\title{Streaming algorithms for recognizing nearly well-parenthesized
expressions}
\author{Andreas Krebs \inst{1}
	\and
	Nutan Limaye  \inst{2}
	\and
	Srikanth Srinivasan\inst{3}\thanks{Supported by NSF grants DMS-0835373}}
\institute{
University of T\"ubingen, Germany. \email{mail@krebs-net.de}
\and
Indian Institute of Technology, Bombay, India. \email{nutan@cse.iitb.ac.in}
\and
DIMACS Center, Rutgers University, USA. \email{srikanth@dimacs.rutgers.edu}
}
\date{}
\usepackage{amsmath}
\usepackage{algorithmic}
\usepackage{algorithm}
\usepackage{fullpage}
\usepackage{epsfig}
\usepackage{multicol}
\usepackage{amssymb}

\newtheorem{obs}{Observation}

\newcommand{\F}{\mathbb{F}}

\newcommand{\U}{\mathcal{U}}

\newcommand{\Dltk}[1]{\Delta^{\leq k}{(#1)}}
\newcommand{\Choose}[2]{\ensuremath{{#1\choose#2}}}

\newcommand{\Dyckt}{\mbox{\sf Dyck}_2}
\newcommand{\Dyckl}{\mbox{\sf Dyck}_l}
\newcommand{\otDyckt}{1\mbox{-turn-}\mbox{\sf Dyck}_2}
\newcommand{\otDyckl}{1\mbox{-turn-}\mbox{\sf Dyck}_l}

\newcommand{\NCo}{\mbox{\sf NC}$^1$}
\newcommand{\Ham}[1]{\mbox{\sf Ham}_{{#1}}}

\newcommand{\PHam}[1]{\mbox{\sf PHam}_{{#1}}}

\newcommand{\IND}{\mbox{\sf IND}$^a_{\mathcal{U}}$}

\newcommand{\naturals}{\mathbb{N}}
\newcommand{\prob}[2]{\mathop{\mathrm{Pr}}_{#1}[#2]}

\begin{document}
\maketitle
\bibliographystyle{plain}

\begin{abstract}
We study the streaming complexity of the membership problem of $\otDyckt$ and
$\Dyckt$ when there are a few errors in the input string. \\

{\bf $\otDyckt$ with errors:} We prove that there exists a randomized one-pass
algorithm that given $x$ checks whether there exists a string
$x' \in \otDyckt$ such that $x$ is obtained by flipping at most $k$ locations of
$x'$ using: 
\begin{itemize}  
\item $O(k \log n)$ space, $O(k \log n)$ randomness, and $poly(k \log n)$ time
per item and with
error at most $1/n^c$. 
\item $O(k^{1+\epsilon} + \log n)$ space for every $0 \leq \epsilon \leq 1$, $O(\log n)$
randomness, $O((\log^{O(1)}n + k^{O(1)}))$ time per item, with error at most $1/8$.
\end{itemize}
Here, we also prove that any randomized one-pass algorithm that makes error at most $k/n$
requires at least $\Omega(k \log(n/k))$ space to accept strings
which are exactly $k$-away from strings in $\otDyckt$ and to reject strings
which are exactly $k+2$-away from strings in  $\otDyckt$. 
Since $\otDyckt$ and the Hamming Distance problem are closely related we also
obtain new upper and lower bounds for this problem.\\

{\bf $\Dyckt$ with errors:} We prove that there exists a randomized one-pass
algorithm that given $x$ checks whether there exists a string
$x' \in \Dyckt$ such that $x$ is obtained from $x'$ by changing (in some
restricted manner) at most $k$ positions using: 
 
\begin{itemize}  
\item $O(k \log n + \sqrt{n \log n})$ space, $O(k \log n)$ randomness,
$poly(k \log n)$ time per element and with error at
most $1/n^c$.
\item $O(k^{1+\epsilon}+ \sqrt{n \log n})$ space for every $0 \leq \epsilon \leq
1$, $O(\log n)$ randomness, $O((\log^{O(1)}n + k^{O(1)}))$ time per element, 
with error at most $1/8$.
\end{itemize}

\end{abstract}
\section{Introduction}
The data streaming model was introduced in the seminal work of Alon et al.
\cite{Alon96}. This model naturally arises in situations where the input data is
massive and rereading the input bits is expensive. The main parameters that play
a role in designing algorithms in such situations are: the space used by the
algorithm, and the number of passes made over the input. An algorithm is said to
be an efficient data streaming algorithm, if the space used by the algorithm is
substantially lesser than the length of the input (sublinear in the length of
the input) and the number of passes is independent of the length of the input. 
Many variants of this basic model have been studied. (See for
example \cite{muthu05} for a survey.)

The membership testing for well-paranthesises strings has been considered in the
past. We denote the set of words with balanced parentheses of
$l$ different types by $\Dyckl$.
It is known that there is a $O(\log n)$ space deterministic algorithm for
testing membership in $\Dyckl$. (In fact the problem is known to be in {\sf
TC}$^0$ \cite{BC89}.) The problem has been considered from property testing
perspective (see for example \cite{AKNS99}, \cite{PRR01}). 
Recently, the problem was considered in the streaming model by Magniez et al.
\cite{MMN09}. It was proved that there is a randomized one-pass streaming
algorithm that takes space $O(\sqrt{n \log n})$ and tests membership in
$\Dyckl$. They also gave an efficient $O(\log
^2n)$ space algorithm which makes bidirectional pass (one forward and one
backward pass) on the input. They also proved a lower bound of $\Omega(\sqrt{n})$ for any
randomized streaming algorithm that makes a single unidirectional (only forward) pass. 
Chakrabarti et al. \cite{CCKM11} and Jain et al. \cite{JainN10} considered the
lower bound problem for unidirectional multi-pass randomized algorithms. In
\cite{CCKM11} it was proved that any $T$-pass (all passes made in the same
direction) randomized algorithm requires $\Omega(\sqrt{n}/T\log \log n)$ space.
Whereas \cite{JainN10} proved $\Omega(\sqrt{n}/T)$ space lower bound
for the same. In \cite{BLV10} membership testing for other classes of languages
was considered. In \cite{BLRV11} it was proved that any randomized $T$ pass
algorithm (passes made in any direction) for testing membership in a
deterministic context-free language requires $\Omega(n/T)$ space. 

We consider a slightly general version of the membership testing problem for
$\Dyckl$.  
Let $\Sigma_l$ denote a set of $l$ pairs of matching parentheses.
We say that an opening parenthesis is \emph{corrupted} if it is
replaced by another opening parenthesis. 
Similarly, a closing
parenthesis is \emph{corrupted} if it is replaced by another closing
parenthesis. For a language $L \in \Sigma_l^*$, let $\Dltk{L}$ be
defined as the set of words over $\Sigma_l^*$ obtained by corrupting
at most $k$ indices of any word in $L$.
In this paper, we consider the membership problem for $\Dltk{\Dyckl}$ and
$\Dltk{\otDyckl}$, where 
$$\otDyckt = \{w\overline{w}^R~|~ w \in \{(, [\}^n~ n\geq 1\}$$
Here, $\overline{w}$ is the string obtained from $w$ by replacing an
opening parenthesis by its corresponding closing parenthesis and $w^R$ is
the reverse of $w$.

Accepting strings with at most $k$ errors is a well-studied problem in many
models of computation. In the streaming model, the problem has been studied in
the past (see for example Cao et al. \cite{CEQZ06}). But we believe that the
problem needs further investigation; this being the primary goal of this paper.

We observe that the membership testing problem for $\Dltk{\Dyckl}$ ($\Dltk{\otDyckl}$)
reduces to the membership testing problem of $\Dltk{\Dyckt}$ ($\Dltk{\otDyckt}$,
respectively). %
We give a simple fingerprinting algorithm
for $\Dltk{\otDyckt}$ that uses $O(k\log n)$ bits of space and randomness. 
The space requirements of this algorithm  are nearly
optimal (because of a communication complexity lower bound of
\cite{HSZZ06}) but the randomness requirements are not. 
We consider the question of derandomizing the above.  The question of
derandomizing streaming algorithms has been considered in the
past (see for example
\cite{Ganguly08},\cite{Nisan90},\cite{RU10},\cite{Shaltiel09}). We
show that the algorithm can be modified to work with just $O(\log n)$
bits of randomness, incurring a small penalty in the
amount of space used. We then consider similar questions for the more
general problem $\Dltk{\Dyckt}$. 
The following table summarizes our algorithmic results:
\smallskip

\noindent\begin{tabular}{|l|l|l|l|l|l|}
\hline
& One-pass & & & &\\
Problem& Algorithm & Space & Randomness & Error & Time (per element)\\
\hline
 & ~~~~1 &
$O(k \log n)$ & $O(k \log n)$ & %
inverse poly
& $poly(k \log n)$ \footnotemark \\
$\otDyckt$ & ~~~~2 &for all $0 < \epsilon < 1$:& & & \\
& & $O(k^{1 + \epsilon} +\log n )$ & $O( \log n )$ & $1/8$ &
$O((\log n)^{O(1)} + k^{O(1)})$  %
\\
\hline
$\Dyckt$ & ~~~~3 &
$O(k \log n + \sqrt{n\log n})$ & $O(k \log n)$ & inverse poly & $poly(k \log n)$
 \footnotemark[\value{footnote}]\\
& ~~~~4 & for all $0 < \epsilon < 1$: & & & \\
& & $O(k^{1+ \epsilon} + \sqrt{n\log n})$ & $O(
\log n)$ & 1/8 & $O( (\log n)^{O(1)} + k^{O(1)}) $  
\\
\hline
\end{tabular}
\footnotetext{In the case of $\Dltk{\otDyckl}$, this is
the exact time per item. However, for $\Dltk{\Dyckl}$ it is the time per item on
average. In the latter case, the algorithm first reads a block and then uses
$O(poly(k \log n ))$ time per element of the block. Therefore, the time per block
is $O(poly(k\log n)\sqrt{n/\log n})$. Both algorithms use an extra post-processing
time of $n^{k+O(1)}$.}
\smallskip

In all the algorithms in the table above, we
assume that the length of the input stream is known.

Using %
Algorithm 1, we can deduce the number
of errors as well as their locations. 
Using a combination of the algorithm for membership testing of $\otDyckt$ due to
\cite{MMN09} (which
we refer to as MMN algorithm) and Algorithm 1, it is easy to get a membership
testing algorithm for $\Dltk{\Dyckt}$. However, such an algorithm  uses
$O(k\sqrt{n \log n})$ space. In order to achieve the claimed bound, we modify
their algorithm for testing membership in $\Dyckt$ and use that in conjunction
with Algorithm 1. In our algorithm, we do not need
to store the partial evaluations of polynomials on the stack.  

Algorithms 2 and 4 are inspired by the
communication complexity protocols of Yao \cite{Yao03} and Huang et
al. \cite{HSZZ06}.  A mere combination of their ideas, however, is 
not enough to get the required
bounds. The crucial observation here is that Yao's protocol can be
derandomized by using a combination of small-bias distributions and
distributions that fool DNF formulae. As this requires very few
random bits, we get the desired derandomization. These algorithms are also
better as compared to Algorithm 1 and 3 in terms of their time complexity.  
For Algorithm 2, we first prove that it suffices to give an efficient algorithm
for $\Ham{n,k}$, where $\Ham{n,k}(x,y)$ for $x,y \in \{0,1\}^n$ is $1$ if and
only if the Hamming distance between $x$ and $y$ is at most $k$. %

Finally, we consider the question of optimality. We prove
that any algorithm that makes $k/n$ error requires $\Omega({k\log
(n/k)})$ space to test membership in $\Dltk{\otDyckt}$ by proving a lower bound
on $\Ham{n,k}$. The two problems are related as follow:
Let $w \in
\Sigma^{2n}$ and let $w =uv$ where $u \in \{(,[\}^n$ and $v \in
\{),]\}^n$. If $($ and $)$ are both mapped to
$0$ and $[$ and $]$ are both mapped to $1$ to obtain $x,y$ from $u,v$ then it
is easy to see that 
$uv \in \Dltk{\otDyckt}$ if and only if $\Ham{n,k}(x,y)=1$.

The problem $\Ham{n,k}$ was considered in \cite{Yao03}, \cite{HSZZ06},
in simultaneous message model. In \cite{HSZZ06}, a lower
bound (in fact, a quantum lower bound) of $\Omega(k)$ was proved for the problem. Their lower
bound holds even for constant error protocols. To best of our knowledge no better lower
bound is known for the problem. We improve on their lower bound by a $\log(n/k)$
factor under the assumption that the communication protocol is allowed to make
small error. Our lower bound can be stated as follows:

\begin{theorem}
\label{thm:lb}
Given two strings $x,y \in \{0,1\}^n$ such that either the Hamming
distance between $x,y$ is exactly $k$ or exactly $k+2$, any randomized
one-pass algorithm that makes error $k/n$ requires space $\Omega(k
\log (n/k))$  to decide which one of the two cases is true for the
given $x,y$ pair.
\end{theorem}

For the lower bound, we use the result of Jayram et al. \cite{JW11}. Intuitively, the
hardest case seems to be to distinguish between exactly $k$ and
exactly $k+2$ errors. The main advantage of our lower bound proof is that it
formalizes this intuition. Moreover, as our algorithm in
Section 3 shows, this bound is tight up to a constant factor for
$n\geq k^2$ (indeed, for $n\geq k^{1+\epsilon}$ for any $\epsilon >
0$).  This bound is not tight in all situations though, for example
when $n\gg k$ but the error is constant. Also, it does not apply to
multi-pass algorithms.  However, this is better than the earlier
bounds \cite{HSZZ06} by a factor of $\log(n/k)$ for small error. %

The rest of the paper is organized as follows: in the next section we give some
basic definitions which will be used later in the paper. In Section
\ref{section_1turn} we give the two randomized one-pass algorithms for testing
membership in $\Dltk{\otDyckt}$. In \ref{dycktsection} we discuss our results
regarding testing membership in $\Dltk{\Dyckt}$. Our lower bound result is
presented in Section \ref{lbsection}.

\section{Definitions and Preliminaries}
\label{section_def}
\newcommand{\poly}{\mathrm{poly}}

\subsection{$\ell$-wise independent hash families}

\begin{definition}
	Given positive integers $\ell,n,m,$ and $s$, a function $F:\{0,1\}^s
	\rightarrow [m]^n$ is an $\ell$-wise independent hash family if
	given any \emph{distinct} $i_1,i_2,\ldots,i_\ell\in [n]$ and
	any (not necessarily distinct) $j_1,j_2,\ldots,j_\ell\in [m]$, we
	have
	\[
	\prob{r\in \{0,1\}^s}{F(r)(i_1) = j_1 \wedge F(r)(i_2) = j_2 \wedge
	\cdots \wedge F(r)(i_\ell) = j_\ell} = \frac{1}{m^\ell}
	\]
	where $F(r)$ is interpreted as a function mapping $[n]$ to $[m]$ in
	the obvious way.
\end{definition}

\begin{lemma}\cite{Vadhan}
	\label{lemma_l_wise}
	For any $\ell,n,m$, there is an $\ell$-wise independent hash family
	$F:\{0,1\}^s\rightarrow [m]^n$, with $s = O(\ell\log(n+m))$ with the
	property that there is a deterministic algorithm which, on input
	$r\in \{0,1\}^s$ and $i\in [n]$, computes $F(r)(i)$ in
	time $\poly(s)$ using space $O(s)$.
\end{lemma}

\subsection{Some pseudorandom distributions for restricted tests}

Given $m\in \mathbb{N}$, we will denote by $\mathcal{U}_m$ the uniform
distribution on $\{0,1\}^m$.

\begin{definition}
Given any class $\mathcal{F}$ of boolean functions defined on
$\{0,1\}^m$, distributions $D_1,D_2$ over $\{0,1\}^m$ and $\delta\in
[0,1]$, we say that \emph{$D_1$ $\delta$-fools $\mathcal{F}$ w.r.t.
$D_2$} if for all $f\in\mathcal{F}$, we have
\[
\left|\prob{z\sim D_1}{f(z) = 1} - \prob{z\sim D_2}{f(z)=1}\right|\leq
\delta
\]
\end{definition}

The above concept has been widely studied for many classes of
functions, especially in the case when $D_2$ is the uniform
distribution $\mathcal{U}_m$. When $D_2$ is the uniform
distribution, in many cases, we know of distributions $D_1$ with very
small support that are nonetheless able to fool some interesting class of functions
$\mathcal{F}$ w.r.t. $D_2$. We note two such results below, since we
will need them later.

\begin{definition}
	Given two vectors $x,y\in\{0,1\}^m$, denote by $\langle x,y\rangle$
	the $F_2$-inner product between $x$ and $y$: that is, $\langle
	x,y\rangle = \bigoplus_i x_iy_i$. For $w\in \{0,1\}^m$, define the
	function $L_w:\{0,1\}^m\rightarrow\{0,1\}$ as follows: $L_w(x) =
	\langle w,x\rangle$. The class of \emph{Linear tests on
	$\{0,1\}^m$} is defined to be the class of functions $\{L_w\ \mid\
	w\in\{0,1\}^m\}$.
\end{definition}

\begin{lemma}[Small bias spaces]\cite{AGHP}
	\label{lemma_small_bias}
	Given any $\delta\in\mathbb{R}^{>0}$ and $m\in\mathbb{N}$, there exists
	an explicit function $G_1:\{0,1\}^s\rightarrow\{0,1\}^m$ for $s =
	O(\log(m/\delta))$ such that the distribution $G_1(r)$ for a randomly
	chosen $r\in\{0,1\}^s$ $\delta$-fools the class of Linear Tests
	w.r.t. the uniform distribution $\mathcal{U}_m$. Moreover, there is
	a deterministic algorithm $\mathcal{A}$ that, given $r\in\{0,1\}^s$
	and $i\in[m]$, computes the $i$th bit of $G_1(r)$ in time $\poly(s)$
	and space $O(s)$.
\end{lemma}

The existence of the algorithm $\mathcal{A}$ as stated in Lemma
\ref{lemma_small_bias} is not formally stated in \cite{AGHP} but
easily follows from Construction $3$ of such spaces in the
paper. We call a $G_1$ as described above a \emph{$\delta$-biased
space} over $\{0,1\}^m$. The distribution $G_1(r)$ for a randomly
chosen $r$ is said to be a \emph{$\delta$-biased distribution}.

The second class of tests we will need to fool is the class of
read-once DNF formulae over $\{0,1\}^m$. It has been proved recently
that $\delta'$-biased distributions for small enough $\delta'$ can be
used to $\delta$-fool the class of read-once DNFs w.r.t. the uniform
distribution.

\begin{lemma}[Fooling read-once DNFs]\cite{DETT}
	\label{lemma_read_once}
	Given any $\delta\in \mathbb{R}^{>0}$ and $m\in\mathbb{N}$, any
	$\delta'$-biased distribution $\delta$-fools the class of read-once
	DNFs over $\{0,1\}^m$ w.r.t. the uniform distribution
	$\mathcal{U}_m$, as long as $\delta' \leq
	\frac{1}{m^{O(\log(1/\delta))}}$.
\end{lemma}

By Lemmas \ref{lemma_small_bias} and \ref{lemma_read_once}, we have

\begin{corollary}
	\label{corollary_read_once}
	Given any $\delta\in\mathbb{R}^{>0}$ and $m\in\mathbb{N}$, there exists
	an explicit function $G_2:\{0,1\}^s\rightarrow\{0,1\}^m$ for $s =
	O(\log m\log(1/\delta))$ such that the distribution $G_2(r)$ for a randomly
	chosen $r\in\{0,1\}^s$ $\delta$-fools the class of Linear Tests
	w.r.t. the uniform distribution $\mathcal{U}_m$. Moreover, there is
	a deterministic algorithm that, given $r\in\{0,1\}^s$ and $i\in[m]$, computes the
	$i$th bit of $G_2(r)$ in time $\poly(s)$ and space $O(s)$.
\end{corollary}

\section{Equivalence with errors}
\label{section_1turn}
In this section, we consider the problem of testing membership in
$\Dltk{\otDyckl}$.  Magniez et al. \cite{MMN09}, showed that it
suffices to design efficient streaming algorithms for testing
membership in $\Dyckt$ in order to get efficient streaming algorithms
for testing membership in $\Dyckl$. Formally,
\begin{lemma}[\cite{MMN09}]
\label{lem:l-to-2-mmn}
If there is a one-pass streaming algorithm for testing membership in
${\Dyckt}$ that uses space $s(n)$ for inputs of length $n$, 
then there is a one-pass streaming algorithm for testing membership in
${\Dyckl}$ that uses space $O(s(n\log l))$ for inputs of length $n$.
\end{lemma}

We first prove a lemma similar to Lemma \ref{lem:l-to-2-mmn}, to state that it
suffices to design
an efficient streaming algorithm for $\Dltk{\otDyckt}$ ($\Dltk{\otDyckt}$) in
order to get an efficient streaming algorithms for $\Dltk{\otDyckl}$
(respectively, $\Dltk{\Dyckl})$. %

\begin{lemma}
\label{lm:dyckl-to-dyckt}%
If there is a one-pass streaming algorithm for testing membership in
$\Delta^{\leq 2k}({\otDyckt})$ ($\Delta^{\leq 2k}({\Dyckt})$) that uses space
$s(n)$ for inputs of length $n$, then there is a streaming algorithm for testing membership in
$\Dltk{\otDyckl}$ ($\Dltk{\Dyckl}$) that uses space $O(s(nl))$ for inputs of
length $n$.
\end{lemma}
\begin{proof}
We use a distance preserving code for this. We encode an opening parenthesis of type $(_i$
by $(^{i-1}~[~(^{l-i}$. And we encode a closing parenthesis of type $)_i$
by $)^{l-i}~]~)^{i-1}$. Now given a string $w \in \Sigma_l^n$, the new string $w'$ is over the alphabet $\Sigma=
\{(,[,),]\}$. And $|w'|=nl$. Also for every mis-match in $w$, $w'$ has two mis-matches. 
Thus the lemma.
\qed\end{proof}

Let $D^{\leq k}(\otDyckt)$ be the set of
string obtained by changing at most $k$ symbols of words in $\otDyckt$. 
Assuming that the length of the string is known,
the membership testing for $D^{\leq k}(\otDyckt)$ (which is more general than
$\Dltk{\otDyckt}$) can also be handled by the
techniques introduced in the paper. If the input string has opening parenthesis
in the first half of the string, then it is considered to be an error. It is
easy to keep track of such errors.

We now note that $\Dltk{\otDyckt}$ on inputs of length $n$ 
reduces to the problem $\Ham{n/2,k}$. 
\begin{lemma}
	\label{lemma_dyck_ham}
	There is a deterministic one-pass streaming algorithm that uses space
	$O(\log n)$ and time $O(n)$, which given as input a string $w\in\{(,[,),]\}^n$,
	outputs a pair of strings $x,y\in\{0,1\}^{n/2}$ and either accepts
	or rejects. If the algorithm rejects, we have
	$w\not\in\Dltk{\otDyckt}$. Otherwise, we have $\Delta(x,y^R)\leq k$
	iff $w\in\Dltk{\otDyckt}$.
\end{lemma}
\begin{proof}
	Given an input of length $n$, the algorithm scans its input from left to
	right, and outputs $0$ on seeing $``("$ and $1$ on seeing $``["$ in
	the first $n/2$ symbols of its input $w$; similarly, on the second
	half of $w$, the algorithm outputs $0$ on seeing $``)"$ and $1$ on
	seeing $``]"$. The algorithm rejects either if it sees the closing
	braces in the first half of its input or the opening braces in the
	second half of its input (in this case, an opening brace has been
	corrupted by a closing brace or vice versa) and accepts otherwise.
	If the algorithm accepts, we see that $\Delta(x,y^R)$ is exactly the
	distance of the input from a string in $\otDyckt$. The lemma
	follows.
\qed\end{proof}

The above lemma shows that it suffices to come up with a streaming
algorithm for the Hamming distance problem to solve the problem
$\Dltk{\otDyckt}$. Once we have such an algorithm, we simply run the
above reduction on an input $w\in\{(,[,),]\}^n$, and obtain strings
$x,y^R$, which we feed in as input to the algorithm for
$\Ham{n/2,k}$ (of course, if the reduction rejects, we reject the
input). Though $\Ham{n,k}$ is only a minor restatement of
$\Dltk{\otDyckt}$, we prefer to work with this problem because of its
cleaner definition.

\begin{theorem}
\label{thm:ham}
For any $k$ and any constant $c>0$, there is a one-pass
randomized streaming algorithm which, when given as input strings
$(x,y^R) \in \{0,1\}^n\times \{0,1\}^n$, that accepts with probability
$1$ if $\Delta(x,y) \leq k$ and rejects with probability $1 - 1/n^c$
if $\Delta(x,y)>k$.  The algorithm also detects the locations where
$x$ and $y$ differ with probability at least $1 - 1/n^c$ if
$\Delta(x,y)\leq k$.  The algorithm uses $O(k\log n)$ space and $O(k \log
n)$ randomness. The time required by the algorithm is $poly(k\log n)$ per item
plus $n^{k+O(1)}$ for post-processing.
\end{theorem}

\begin{proof}
The algorithm uses a fingerprinting strategy and is directly inspired
by the standard randomized communication complexity protocol for the
Equality problem (see \cite{KN06}, for example). Fix a field
$\F_{2^\ell}$, where the exact value of $\ell$ will be determined
later. We call a polynomial $p(z)\in \F_{2^\ell}[z]$ \emph{boolean} if
all of its coefficients are $0$ or $1$. The \emph{weight} of a boolean
polynomial $p$ will be the number of non-zero coefficients of $p$.

We think of $w\in\{0,1\}^n$ as defining a boolean polynomial
$p_w(z)\in\F_{2^\ell}[z]$ as follows: $p_w(z) = \sum_{i=1}^n w_i
z^{i-1}$, where $w_i$ denotes the $i$th bit of $w$. Note that the
polynomial $q_{x,y}(z) := p_x(z) + p_y(z)$ is a boolean polynomial of
weight exactly $\Delta(x,y)$. We check that $\Delta(x,y)\leq k$ by
evaluating $q_{x,y}(z)$ at a random $\alpha\in\F_{2^\ell}$. More
formally, the algorithm is:
\begin{itemize}
	\item Pick $\alpha\in\F_{2^\ell}$ uniformly at random.
	\item Check if $q_{x,y}(\alpha) = p(\alpha)$ for any boolean
		polynomial $p$ of degree less than $n$ and weight at most $k$. If
		not, REJECT.  
	\item If the above does hold for \emph{some} boolean polynomial of weight at
		most $k$, ACCEPT and pick \emph{any} such polynomial $p(z) =
		\sum_{i}p_i z^i$. Let $S =\{i\ \mid\ p_i\neq 0\}$ be the support
		of $p$.  Output $S$ as the estimate of points where $x$ and $y$
		differ.
\end{itemize}

Let us first establish the correctness of the above algorithm
(assuming $\ell$ is large enough). Clearly, if $\Delta(x,y)\leq k$,
then $q_{x,y}(z)$ is a polynomial of weight at most $k$ and the
algorithm always accepts. The algorithm can only err if: (a)
$\Delta(x,y) > k$ or (b) $\Delta(x,y) \leq k$ but the algorithm
outputs the wrong set of indices as its estimate of where $x$ and $y$
differ. In either case, there is a boolean polynomial $p(z)$ of degree
less than $n$ and weight at most $k$ such that $q_{x,y}(z)\neq p(z)$
but $q_{x,y}(\alpha) = p(\alpha)$. For any fixed polynomial $p(z)$,
this happens with probability at most $n/2^\ell$ by the
Schwartz-Zippel Lemma. Since the number of polynomials of weight at
most $k$ is at most $n^k$, the probability that there exists
\emph{any} such polynomial $p$ is bounded by $n^{k+1}/2^\ell$.
Choosing $\ell = O(k\log n)$, we can reduce this error to $1/n^c$ as
claimed. 

Computing $q_{x,y}(\alpha)$ can easily be done in a one-pass fashion
using space $O(\ell) = O(k\log n)$ and time $poly(k \log n)$ per item. 
After reading the stream, we need to cycle through the $n^k$ boolean 
polynomials $p$ of weight at most $k$ and compute the values they take
at input $\alpha\in\F_{2^\ell}$, which can also be done in space
$O(k\log n + \ell) = O(k\log n)$ and time $n^k\poly(k\log n) =
n^{k+O(1)}$, as claimed above. This completes the proof of the theorem.
\qed\end{proof}

\subsection{A randomness-efficient streaming algorithm for Hamming Distance}

Above, we showed that $\Ham{n,k}$ can be computed using space
$O(k\log n)$ and $O(k\log n)$ random bits. Are these parameters
optimal? As we will show later in Section \ref{lbsection}, the bound
on space is nearly optimal. However, we show in this section that the
number of random bits can be significantly reduced, if one is willing
to use a small amount of additional space. The ideas in this section
go back to the results of Yao \cite{Yao03} and Huang et
al. \cite{HSZZ06}, who designed efficient randomized communication
complexity protocols for the two-party problem of checking if the
Hamming Distance between $x$ and $y$ is at most $k$.

Let $\PHam{n,k,l}: \{0,1\}^n\times \{0,1\}^n \rightarrow \{0,1\}$ be a
partial function, defined as follows: 
On input $(x,y^R)$ it evaluates to $0$ if the hamming distance between $x$ and $y$ is greater than
or equal to $l$, it evaluates to $1$ if the distance is less than or equal to $k$ and 
is not defined on other inputs.

\begin{theorem}
	\label{thm_ham_rand}
For every constant $0 \leq \epsilon \leq 1$ there is a randomized
one-pass streaming algorithm that computes $\Ham{n,k}$ using
$O(k^{1+\epsilon}+\log n)$ space and $O(\log n)$ randomness and errs
with probability bounded by $1/8$. The time taken by the algorithm is
$O(\log n)^{O(1)} + k^{O(1)})$ per item.
\end{theorem}

\noindent
\textbf{Proof Strategy}: In order to prove the above theorem, we divide the problem into two parts. 
Assuming there are at most $2k$ errors, we design an algorithm that computes $\Ham{n,k}$ correctly
with high probability. We call this the inner algorithm. 
We design another randomized algorithm to compute $\PHam{n,k,2k}$
with high probability. We call this algorithm the outer algorithm. 

We output $1$, that is we declare that the number of 
indices where $x$ and $y$ differ is at most $k$, if and only if both the inner and the outer algorithms 
output $1$. If $x$ and $y$ differ on more than $2k$ indices, then the outer algorithm
will output $0$ with high probability. The answer of the inner algorithm will not be reliable in this case. 
Where as if they differ on more than $k$ but less than $2k$ places then the inner algorithm will output
$0$ with high probability. Let $\gamma_1,\gamma_2$ be errors in inner and outer algorithms respectively.
Then the overall error $\gamma$ is bounded by $\gamma_1+\gamma_2$. 
We prove that both $\gamma_1$ and $\gamma_2$ are bounded by $\gamma/2$
for a small constant $\gamma$.

\subsubsection{Inner algorithm} 

\begin{definition}
	Given, $k,n\in\naturals$, we say that an element $w\in
	([k]\times\{0,1\})^n$ is an \emph{XOR representation of length $n$}
	of a string $a\in \{0,1\}^k$ if for
	each $j\in [k]$, we have $a_j = \bigoplus_{i: w_i = (j,u_i)}u_i$.
\end{definition}
We think of the XOR representation as streaming updates of a structure over
$\F_2$.
\begin{lemma}
	\label{lemma_inner_1}
There is a randomized one-pass streaming algorithm which given input 
$x,y^R \in \{0,1\}^n$ such that $\Delta(x,y) \leq 2k$ computes an XOR
representation of length $n$ of $a,b \in \{0,1\}^{16k^2/\gamma}$ such that with
probability $1-\gamma/4$, $\Ham{n,k}(x,y) = \Ham{16k^2/\gamma,
k}(a,b)$ The algorithm uses $O(\log n)$ bits of randomness,
$O(\log n)$ space, and $(\log(n/\gamma))^{O(1)}$ time per item. 
\end{lemma}

\begin{proof}
	The proof is simple. We pick a random hash function $h$ from a pairwise
	independent hash family of functions mapping $[n]$ to
	$[16k^2/\gamma]$. We think of $h$ as dividing the $n$ indices of $x$
	and $y$ into $16k^2/\gamma$ buckets. 

	Given $x,y$ such that $\Delta(x,y)\leq 2k$, call index $i$
	\emph{good} if $x_i\neq y_i$. Given two good indices $i\neq j$, the
	probability that $h$ maps both of them to the same bucket is at most
	$\gamma/16k^2$. A simple union bound tells us that with probability
	$1-\gamma/4$, all the good indices are mapped to different buckets. 

	After having picked $h$, the streaming algorithm computes the XOR
	representations of $a$, $b$ defined as follows: for any $j$, $a_j$
	is the XOR of the bits of $x$ whose indices are in the $j$th bucket;
	formally, $a_j = \bigoplus_{i: h(i) = j}x_i$; the string $b$ is
	similarly related to $y$. Clearly, if $h$ maps the good indices to
	different buckets, then $a_j \neq b_j$ iff the $j$th bucket contains
	a good index and hence $\Delta(a,b) = \Delta(x,y)$. On reading the
	input bit $x_i$, the algorithm computes the bucket $j=h(i)$ and
	writes down $(j,x_i)$ which in effect updates the $j$th bit of $a$.
	In a similar way, when scanning $y$, the algorithm updates $b$.

	The space and randomness requirements are easily analyzed. Picking a
	random hash function $h$ from a pairwise independent family as above requires
	$O(\max\{\log n, \log (k^2/\gamma)\}) = O(\log(n/\gamma))$ random
	bits by Lemma \ref{lemma_l_wise}. The algorithm needs to store these
	random bits only. Computing $h(j)$ for any $j$ only requires space
	$O(\log n/\gamma)$. Finally, the processing time per element is
	$O(\poly(\log(n/\gamma))$.
\qed\end{proof}

We will use the above algorithm as a streaming reduction and solve the
problem using the algorithm of Lemma \ref{lemma_inner_2}.

\begin{lemma}
	\label{lemma_inner_2}
	For any $n,k$ and every constant $0<\epsilon<1$ and $\gamma\geq
	\frac{1}{k^{O(1)}}$, there is a randomized one-pass streaming algorithm
	which, on inputs $a,b\in\{0,1\}^{16k^2/\gamma}$ accepts iff
	$\Delta(a,b)\leq k$ with error probability at most $\frac{\gamma}{4}$. The
	algorithm uses $O(\log k)$ bits of randomness,
	$O(k^{1+\epsilon}+\log n)$ space, and time per element is $k^{O(1)}$. The algorithm
	expects its inputs $a,b$ to be given in terms of XOR
	representations of length $n$.
\end{lemma}
\begin{proof}
First, we present the algorithm for the special case when $n=k$ and
the input is simply the pair of strings $a,b$ in the natural order of
increasing indices. We will then explain the simple modifications that
are necessary for the case when the input is a pair of XOR
representations of length $n$.

Fix a positive constant $\delta < \epsilon$. Let $h: [16k^2/\gamma]
\rightarrow [k^{1+\delta}]$ be a function picked at random.  Let $j
\in [k^{1+\delta}]$ be a fixed bucket. We have $\prob{}{h(i) = j} =
\frac{1}{k^{1+\delta}}$.

Define a set $I$ of indices as follows: if $\Delta(a,b)\leq k$, then
let $I$ be the indices where $a$ and $b$ differ; otherwise, let $I$ be
any set of $k+1$ indices where $a$ and $b$ differ. Let $u$ be the size
of a subset $U$ of $I$. We have $\prob{}{h(U) = j} \leq
\frac{1}{(k^{1+\delta})^u}$. By a union bound over $U$ of size $u$,
$\prob{}{\exists U: h(U) = j} \leq
\frac{\Choose{k+1}{u}}{(k^{1+\delta})^u} \leq
\frac{(k+1)^u}{(k^{1+\delta})^u}\leq \frac{1}{k^{\delta u/2}}$.
Therefore, since there are at most $k^2$ buckets, $\prob{}{\exists U
~\exists \mbox{ a bucket }j:\  h(U) = j} \leq 1/{k^{\delta u/2 - 2}}$.

We want this probability to be less than $\gamma/8$. Therefore we
select $u = O(1/\delta + \frac{\log(1/\gamma)}{\log k}) =
O(1)$, where the last equality uses $\gamma \geq
\frac{1}{k^{O(1)}}$ and $\delta$ is a constant. 

Note that the above argument works if we used a function from a $u$-wise independent
family of functions rather than a random function. This requires only 
$O(u\log(k/\gamma)) =
O(\log k)$ bits of randomness and space $O(\log(k/\gamma))$ by Lemma
\ref{lemma_l_wise}.  Hereafter, we assume that we have picked a hash
function $h$ from this family so that each bucket $j\in[k^{1+\delta}]$
contains at most $u$ indices from $I$. Let $B^a_j$ and $B^b_j$ be the
buckets formed by hashing $a$ and $b$ respectively, where $1 \leq j
\leq k^{1+\delta}$.

Given boolean strings $a',b'$ of the same length, define
$\Delta_u(a',b')$ to be $\min\{\Delta(a',b'),u\}$. We will compute the
function $F(a,b) = \sum_{j\in [k^{1+\delta}]}\Delta_u(B^a_j,B^b_j)$
and accept if the value computed is at most $k$. It can easily be
seen, using the properties of $h$, that this computes $\Ham{16k^2/\gamma, k}(a,b)$.

Computing $\Delta_u(B^a_j,B^b_j)$ for any $j$ is easily done using the
ideas of the algorithm of Theorem \ref{thm:ham}. We work over the
field $\F_{2^\ell}$ where $\ell$ is a parameter that we will fix
shortly. Given a polynomial $p\in \F_{2^\ell}[z]$ with only $0$-$1$
coefficients, we denote by \emph{the weight of $p$} the number of
non-zero coefficients of $p$.  Given $c\in \{a,b\}$ and $j\in
[k^{1+\delta}]$, the bucket $B^c_j$ defines for us the polynomial
$p_{j,c}(z) = \sum_{i\in h^{-1}(j)}c_iz^i$ over $\F_{2^\ell}$. Define the
polynomial $q_j(z) = p_{j,a}(z) + p_{j,b}(z)$.  The weight of $q_j$ is exactly
$\Delta(B^a_j,B^b_j)$. The algorithm to compute $F(a,b)$ is the
following:

\begin{enumerate}
	\item Pick $\alpha\in\F_{2^\ell}$ uniformly at random.
	\item For each $j$, compute $q_j(\alpha)$ and check if it evaluates
		to the same value as some polynomial $p$ of weight at most $u$. If
		so, let $w_j$ be the weight of an arbitrary such $p$; if not, let
		$w_j = u$.
	\item Output $\sum_j w_j$.
\end{enumerate}

The above algorithm errs on bucket $j$ only if there is a polynomial
$p\neq q_j$ of weight at most $u$ such that $p(\alpha) = q_j(\alpha)$.
This occurs with probability at most $O(\frac{k^2}{\gamma 2^\ell})$ for a
fixed polynomial $p$ and hence with probability at most
$\frac{(k/\gamma)^{O(u)}}{2^\ell}$ for \emph{some} polynomial $p$ of
weight at most $u$ after a union bound. After taking a union bound
over buckets, we get a failure probability of at most
$\frac{(k/\gamma)^{O(u)}}{2^\ell}$ (with different constants in the
exponents). Choosing $\ell = O(u\log(\frac{k}{\gamma})) =
O(\log k)$, we can reduce
the error to $\gamma/8$.

The overall error of the algorithm is at most $\gamma/4$. The space
used per bucket and the number of random bits used is at most
$O(\log k)$. Adding up over all buckets, the
space used is bounded by $O(k^{1+\delta}\log k)$, which is at most
$O(k^{1+\epsilon})$. The time taken by the algorithm to compute the
values $\{q_j(\alpha)\ \mid\ j\in[k^{1+\delta}]\}$ is $k^{O(1)}$.
Finally, checking if each $q_j(\alpha)$ evaluates to the same value as
a polynomial of weight at most $u$ takes time $k^{O(u)} = k^{O(1)}$.

Now for the case when the input is given as a pair of XOR
representations of length $n$. We simply note that the polynomials
$p_{j,a}$ and $p_{j,b}$ are still easy to compute. For example, on
reading the $i$th element $w_i = (j_i,u_i)$ of the XOR representation of
$a$, the algorithm simply updates the current value of
$p_{h(j_i),a}(\alpha)$ by adding $u_i\alpha^{j_i}$ to it; this works
as intended since $\F_{2^\ell}$ is a field of characteristic $2$. The
algorithm only needs an additional counter that counts up to $n$ so
that it knows when the XOR representation of $a$ ends.
\qed\end{proof}

Setting $\gamma$ to be $1/8$ in the Lemmas \ref{lemma_inner_1} and
\ref{lemma_inner_2}, we see that the space taken by the Inner
algorithm overall is $O(k^{1+\epsilon} + \log n)$, the amount of
randomness used is $O(\log n)$ and the time taken per item is $O((\log
n)^{O(1)} + k^{O(1)})$.

\subsubsection{Outer algorithm}

Given $x,y\in\{0,1\}^n$, we denote by $\langle x,y\rangle$ the
$\F_2$-inner product of $x$ and $y$. Formally, $\langle x,y\rangle =
\bigoplus_{i=1}^n x_iy_i$. 

\begin{lemma}
There is a randomized one-pass streaming algorithm that computes $\PHam{n,k,2k}$
correctly with probability $1-\gamma/2$ using $O(\log n\log(1/\gamma))$
bits of space and randomness and time per item $(\log
n)^{O(1)}\log(1/\gamma)$.
\end{lemma}

\begin{proof}
For simplicity, we will assume that
$k$ is a power of $2$. All the results carry through in the general
case, with only superficial changes.

We use a protocol of Yao \cite{Yao03}. Yao devised a
one-way\footnote{Actually, Yao's protocol works in the more
restrictive \emph{simultaneous message} model, but this fact is not
relevant here.} randomized communication complexity protocol using
which two players Alice and Bob, given inputs $x$ and $y$
respectively, can decide $\PHam{n,k,2k}(x,y)$
using $O(\log(1/\gamma))$ bits of communication. A brief sketch
follows.  Let $u$ denote $x\oplus y$. Using public randomness, Alice
and Bob pick random strings $z_1,\ldots,z_\ell\in \{0,1\}^n$ for $\ell
= O(\log(1/\gamma))$ such that each bit of each $z_i$ is set to $1$
independently with probability $1/4k$ --- we call this distribution
$D_{1/4k}$. Alice computes $\langle x,z_1\rangle,\ldots, \langle
x,z_\ell\rangle$ and sends them to Bob who uses them to compute
$\langle u,z_1\rangle,\ldots, \langle u,z_\ell\rangle$. 

Let $z_i$ be picked from $D_{1/4k}$. It is easily checked that if
$\Delta(x,y)\leq k$, then $\langle u,z_i\rangle$ takes value
$1$ with probability at most $p_1 = \frac{(1-1/\sqrt{e} + o_k(1))}{2}$. On the other
hand, if $\Delta(x,y)\geq 2k$, then $\langle u, z_i\rangle$
takes value $1$ with probability at least $p_2 = \frac{(1-1/e + o_k(1))}{2}$.
Thus, by performing a suitable threshold on the number of $i$ such
that $\langle u, z_i\rangle = 1$ --- say by checking if the
number of $1$s is at least $\frac{p_1+p_2}{2}$ --- Bob can compute
$\PHam{n,k,2k}(x,y)$ correctly with error probability at most
$\frac{\gamma}{2}$.

Using the above ideas, we wish to come up with a streaming algorithm
for this problem that is also randomness efficient. Both these
constraints require us to change the original protocol. (Note that the
obvious implementation of the above protocol in the streaming setting
will require $\Omega(n\ell)$ bits of space and randomness.) To reduce
the amount of randomness, we run Yao's protocol with pseudorandom $z$
from a distribution $D$ that fools the linear test defined by the
string $u$ (see Section \ref{section_def} for the definition of a
``linear test''). Formally, we want a distribution $D$ over $\{0,1\}^n$
such that for
any $w\in \{0,1\}^n$ 

\[ \left| \prob{z\sim
D}{\langle w,z\rangle = 1} - \prob{z\sim D_{1/4k}}{\langle w,z\rangle
= 1}\right|\leq
1/100
\]

That is, we want a distribution $D$ that $1/100$-fools linear tests
w.r.t. $D_{1/4k}$. Furthermore, we would like to be able to sample
from $D$ using a small number of random bits.

Fix $w\in \{0,1\}^n$. Let $f_w$ be the following related test, defined
on $\{0,1\}^{nt}$:
$$f_w(z'_{11},z'_{12},\ldots,z'_{1t},z'_{21},z'_{22},\ldots,z'_{2t}, \ldots,
z'_{n1},z'_{n2},\ldots,z'_{nt}) = \bigoplus_{i: w_i = 1}(\bigwedge_{j=1}^{l}
z'_{ij})$$
where $2^t = 4k$. Note that 
$\prob{z \in D_{1/4k}}{\langle w,z\rangle = 1} = \prob{z' \in
\U_{nt}}{f_w(z')=1}$,
where $\U_r$ denotes the uniform distribution on $\{0,1\}^r$.

We will first design a distribution $D'$ over $\{0,1\}^{nt}$ that
$1/100$-fools the family of tests $\{f_w\ |\ w\in\{0,1\}^n\}$ w.r.t.
the uniform distribution $\U_{nt}$. 

Now, we describe $D'$. We break the variables $\{z'_{ij}\}$ into two
blocks: $B_1 := \{z'_{i1}\ |\ i\in [n]\}$ and $B_2 := \{z'_{ij}\ |\
i\in [n], j\neq 1\}$. Consider the test $f_w$ evaluated at a random
input $z'\in\{0,1\}^{nt}$. It is helpful to view
this evaluation as a two-step process: In the first step, we
substitute a uniform random string $z'_2$ for the tuple of variables
in $B_2$. After this substitution, $f_w(\cdot, z'_2)$ becomes a linear function on
the variables in $B_1$. If this linear function is the zero linear
function, then $f_w(z'_1,z'_2)$ cannot evaluate to $1$ on any setting
$z'_1$ of the variables in $B_1$. On the other hand, if this linear
function is non-zero, then $f_w(z'_1,z'_2)$ evaluates to $1$ with
probability exactly $1/2$ over the choice of $z'_1$. Putting things
together, we see that $\prob{z'}{f_w(z') = 1} =
(1/2)\prob{z'_2}{g_w(z'_2) = 1}$, where $g_w$ is the following
\emph{read-once} DNF formula on the variables in $B_2$ that tells us exactly
when $f_w$ becomes a non-zero linear function on the variables in
$B_1$: $g_w = \bigvee_{i: w_i = 1}\bigwedge_{j=2}^{l} z'_{ij}$

This tells us that we only need to fool read-once DNFs and linear
tests w.r.t. the uniform distribution to fool $f_w$-tests w.r.t. the
uniform distribution. We will generate $z'\in D'$ as follows: $z'_1$
will be sampled from an explicit $\delta_1$-biased space $D'_1$ and
$z'_2$ will be independently sampled from an explicit space $D'_2$
that $\delta_2$-fools read once DNFs. We have:
\begin{align*}
	&\left|\prob{z'\sim D'}{f_w(z') = 1} -
	\prob{z'\sim \U_{nt}}{f_w(z') = 1}\right| \\ 
	&=\left|\prob{z'_2\sim D'_2}{g_w(z'_2) = 1}\cdot \prob{z'_1\sim
	D'_1}{f_w(z'_1,z'_2) = 1\ |\ g_w(z'_2) = 1} -
	\frac{\prob{z'_2\sim\U_{n(t-1)}}{g_w(z'_2) = 1}}{2}\right|\\
	&\leq \left|\prob{z'_2\sim D'_2}{g_w(z'_2) = 1} -
	\prob{z'_2\sim\U_{n(t-1)}}{g_w(z'_2) = 1}\right| + \left|\prob{z'_1\sim
	D'_1}{f_w(z'_1,z'_2) = 1\ |\ g_w(z'_2) = 1} - \frac{1}{2}\right|\\
	&\leq \delta_1 + \delta_2
\end{align*}

where the first inequality uses the fact that $|pq -rs|\leq |p-r| +
|q-s|$ for any $p,q,r,s\in [0,1]$, and the second inequality follows
from the definitions of $D'_1$ and $D'_2$. Choosing $\delta_1$ and
$\delta_2$ to be small enough constants, we obtain a pseudorandom
distribution $D'$ that $1/100$-fools $f_w$-tests w.r.t. the uniform
distribution. By Lemma \ref{lemma_small_bias} and Corollary
\ref{corollary_read_once}, the amount of randomness required for the
above is $O(\log n)$. Using $D'$, we can define a distribution $D$
that fools linear tests w.r.t. $D_{1/4k}$ as follows: to pick $z\sim
D$, we pick $z'\sim D'$ and output $z$ defined by $z_i =
\wedge_{j=1}^t z'_{ij}$ for each $i$. It is easily seen from the
definition of the tests $f_w$ that the distribution $D$ $1/100$-fools
all linear tests w.r.t. $D_{1/4k}$.  Since no additional randomness is
used to generate $D$, the amount of randomness used is $O(\log n)$.
Moreover, by Lemma \ref{lemma_small_bias} and Corollary
\ref{corollary_read_once}, given a random seed $r$ of length $O(\log
n)$ and $j\in [n]$, the $j$th bit of the output of $D$ on this random
seed can be generated using $O(\log n)$ bits of space in time
$\poly(\log n)$.

With the pseudorandom distribution $D$ in place, we can run Yao's
protocol in the streaming setting with independent random strings
$z_1,\ldots,z_\ell$ picked from the distribution $D$. Exactly as
above, for suitably chosen $\ell = O(\log (1/\gamma))$, this algorithm
computes $\PHam{n,k,2k}(x,y)$ with error probability at most
$\gamma/2$. The space and randomness used are both $O(\log n\log
1/\gamma)$ and the time taken is $n\poly(\log
n)\log(1/\gamma)$. Setting $\gamma = 1/8$, this proves the lemma and
also concludes the proof of Theorem \ref{thm_ham_rand}. 
\qed\end{proof}

\section{Accepting $\Dyckt$ with errors}
\label{dycktsection}
In this section we consider the membership problem of $\Dltk{\Dyckt}$. We assume
that disregarding the type of the brackets the string is well-matched. We only
consider the kind of errors where an opening(closing) parenthesis of one type is replaced by an
opening(closing, respectively) parenthesis of the other type. 
We prove the following theorem:
\begin{theorem}
\label{thm:deltak-dyckt}
For any $k$ there exists a constant $c>0$,  
there is a randomized one-pass algorithm such that, given a
string $w \in \Sigma^n$, 
if $w \in \Dltk{\Dyckt}$ then with probability at least $1 -1/n^c$
it accepts
$w$ and if $w \notin \Dltk{\Dyckt}$, then with probability $1 -1/n^c$ it rejects
$w$. The algorithm uses $O(k \log n + \sqrt{n \log n})$ space and takes $poly(k
\log n)$ time per item and $n^{k+O(1)}$ time for post-processing.  
\end{theorem}

It is easy to see that combining the ideas from \cite{MMN09} and from the
previous section, we can accept $\Dltk{\Dyckt}$ using $O(k\sqrt{n \log n})$
space. But for the bound stated in Theorem \ref{thm:deltak-dyckt}, we need more
work.

In \cite{MMN09} a one-pass randomized streaming algorithm was given for the
membership testing of $\Dyckt$. We refer to that as the MMN algorithm. 
We make one change to the MMN algorithm. We use the stack only to store
indices from the set $[n]$, and do not store the partial evaluations of a
polynomial on the stack.

Divide the input into $\sqrt{n /\log n}$ blocks
of length $\sqrt{{n}{\log n}}$ each. In each block, check for balanced
parentheses and if there are less than or equal to $k-Err$ mis-matches, then reduce
the remaining string to a string (possibly empty)
of closing parentheses followed by a string (possibly empty) of opening
parentheses. Here, $Err$ is a counter that maintains the number of 
mismatches found so far. 
If $Err$ exceeds $k$, then halt the algorithm and reject.   

Let $x$ denote the reduced string obtained by matching parentheses within each
block. (Note that this can be done in a streaming fashion.) For the reduced string
$x$ we say that the  
opening parenthesis at position $i$ has an index $h_{x,i}$ if it is the $h$th opening parenthesis in
$x$. We say that the closing parenthesis has index $h_{x,i}$ if it is the
the closing parenthesis that closes an opening parenthesis having index $h_{x,i}$
in $x$.  We drop $x$ when it is clear from the context.

\begin{obs}
\label{obs:index}
Note that no two opening (or two
closing) parentheses have the same index.

Also, an opening parenthesis
has the same index as another closing parenthesis if and only if they form a matching pair
in the string obtained from the input string by disregarding the type of the
parentheses. 

For example in the input $(([])[])$ the indices of the opening parentheses would
be $(^1(^2[^3])[^4])$ and the indices of the opening and closing parentheses
would be $(^1(^2[^3]^3)^2[^4]^4)^1$.  If we reorder the input such that all
opening parentheses are in the first half with ascending index and the closing
parentheses are in the second half with descending input we get
$(^1(^2[^3[^4]^4]^3)^2)^1$. 
\end{obs}

We now describe the procedure for computing the index. 
The index of an opening parenthesis is easy to compute. It is a monotonically
increasing quantity and can be stored in a $O(\log n)$ bit counter say $c_{open}$.
It is initialized to $0$ and incremented by $1$ every time an opening parenthesis
is encountered. To compute the index of a closing parenthesis we use the
stack. 
The stack is only being used to compute the index. At any stage during the
algorithm we maintain the intervals of yet to be matched open parentheses on the
stack. %

The first block consists of only opening parentheses (assuming we have already reduced
the block). After processing the first
block we push $[1,c_{open}]$ on the stack to remember the interval of
indices yet to be matched. 
(In the case that all parentheses of the first block are matched within the
first block, the next block is treated as the first block).

Now suppose we process the next block (also assume that this block is already reduced).  A reduced block will consist of a sequence of closing parentheses (possible empty) followed by a sequence of opening parentheses (possible empty).
Recall we maintain the intervals yet to be matched on the stack, say the stack-top is $[m,m']$. 

If the block begins with a (non-empty) string of closing parentheses, $m'$ is
the index of the first closing parenthesis, and is decremented after reading the
closing parenthesis. As long as $m'\geq m$, this continues. If $m'<m$ we get the
next interval of unmatched parenthesis from the stack. (If the stack is empty the input is not well matched disregarding the type of parentheses and we reject the input).  If the string of closing parentheses ends while $m\geq m'$, we push the remaining interval to the stack.

When reading the string of opening parentheses we let $m=c_{open}$ at the
beginning, process all opening parentheses, and before moving on to the next
block we push $[m, c_{open}]$ (where $c_{open}$ is updated value after reading
the string of opening parentheses).

While processing closing
parentheses, the stack is, if at all, popped but never pushed. While processing
opening parentheses, a new stack element may be pushed. However, this happens at
most 
once per block, and therefore at most $\sqrt{n/\log n}$ times. Also the elements of the stack are
tuples of indices so they are $O(\log n)$ bits. Therefore, the total space needed 
to compute indices is $O(\log n\sqrt{n /\log n})=O(\sqrt{n\log n})$. 

So we can compute the index of the parentheses, and now we show how to use this
to compute $\Dltk{\Dyckt}$.  
Now assume that, at any stage $i$, we have the index of
the input symbol $x_i$.  Let the sign of the opening parentheses be $+1$ and
that of closing parentheses be $-1$. 
We think of the reduced string $x \in \{(,[,),]\}^*$ as a string over $\{0,1\}^*$ by replacing
every occurrence of `(' and `)' by a $0$ and every occurrence of `[' and `]' by a
$1$. We think of this string defining a Boolean polynomial $p_x(z) =
\sum_i sign(x_i) x_iz^{\mbox{\scriptsize{index of }} x_i}.$ 
Due to Observation \ref{obs:index}, it is easy to see that the weight of the
polynomial $p_x$ is at most $k-Err$ if and only if $w \in \Dltk{\Dyckt}$. 
We now check whether $w \in \Dltk{\Dyckt}$ by evaluating $p_x$ at a random
$\alpha \in \F_{2^l}$. Assuming that we know how to compute index of $x_i$ at
step $i$, we can evaluate $p_x$ as in the proof of Theorem \ref{thm:ham}.

Given below is the algorithm that uses the index finding procedure as described
above, and evaluates polynomial $p_x$ at a random location to test membership of
$w$ in $\Dltk{\Dyckt}$. 
In addition to th space required for computing the indices, $O(l)$ bits are
required to store evaluation of $p_x$. But this does not
need to be stored multiple times on the stack. Therefore, the algorithm uses
$O(l + \sqrt{n \log n}) = O(k \log n + \sqrt{n \log n})$ space.

The proof of correctness and the error analysis are similar to the proof of 
Theorem \ref{thm:ham}. Thus we get Theorem \ref{thm:deltak-dyckt}. 
The detailed algorithm is given below.

\bigskip
\hrule
  \begin{algorithmic}[1]
	\STATE pick $\alpha \in_{R} \F_p$, set $sum \leftarrow 0$, set
		$c_{open} \leftarrow 0$, set $Err \leftarrow 0$,
	\FOR{each block}
	\STATE read the word $y$ consisting of the next $\sqrt{n\log n}$ letters
		(or less if the stream becomes empty),	
	\STATE check that matching pairs within $y$ have at most $k-Err$ errors (if not,
		reject: “mismatched parentheses”), if so, update $Err$ to this
value,
	\STATE simplify $y$ into $uv$ , where $u$ has only closing parentheses and $v$
		has only opening parentheses,
	\FOR{$i=1$ to $|u|$}
		\STATE pop $[m,m']$ from the stack (reject if nothing to pop),
		\STATE $sum \leftarrow 
				\begin{array}{cc}
				sum - \alpha^{m'} & \mbox{if } u[i] = `]'
			\mbox{~~~~(unchanged otherwise)}\\
						\end{array}
			$
		\STATE $m' \leftarrow m' - 1$,
		\STATE if $m \geq m'$, push $[m,m']$,
	\ENDFOR
	\STATE $sum \leftarrow sum + \sum_{j=c_{open}}^{c_{open}+|v|-1}
v_j\alpha^{j}$, \newline(by abuse of notation, $v_j=1$ if $v_j=$`[' and is
$0$ otherwise)
	\STATE push $[c_{open}, c_{open}+ |v|]$,
	\STATE set $c_{open} \leftarrow c_{open}+ |v|$,
	\ENDFOR
	\STATE Check $sum=p(\alpha)$ for any polynomial $p$ with $0$-$1$
coefficients of degree less than $n$ and weight at most $k$. If not, REJECT.
  \end{algorithmic}
\hrule
\paragraph{Reducing the randomness}
The ideas used in reducing randomness for $\Dltk{\otDyckt}$ also work for
reducing randomness in the membership
testing of $\Dltk{\Dyckt}$. Here, instead of hashing the input positions, we
hash the indices using the random hash functions. For computing the
indices, we use the procedure described above. Instead of computing polynomials,
we compute hashes and follow the steps as in Sections 3.1, 3.2, and 3.3.

We get the following theorem:
\begin{theorem}
For every constant $0 \leq \epsilon \leq 1$, there is  randomized one-pass
algorithm that tests membership in $\Dltk{\Dyckt}$ using $O(k^{1+\epsilon} +
\sqrt{n\log n})$ space, $O(\log n)$ randomness, $O(\log^{O(1)}n + k^{O(1)})$
time and errs with probability $1/8$.
\end{theorem}  

\section{Lower bound}
\label{lbsection}

We will show a lower bound for $\PHam{n,k,k+2}$ by reducing the augmented indexing problem \IND\ (see \cite{JW11}) to it.

Let $\U \cup \{\perp\}$ denote  a large input domain, where $\perp \notin \U$.
Define the problem \IND\ as follows:
Alice and Bob receive inputs $x=(x_1,x_2,\ldots,x_N) \ \in \U^N$ and
$y=(y_1,y_2,\ldots,y_N) \in (U\cup \{\perp\})^N$, respectively. 
The inputs have the following promise:
There is some unique $i\in[N]$, such that $y_i \in \U$, and for $k<i$: $x_k=y_k$, and for $k>i$: $y_k = \perp$. 
The problem \IND\ is defined only over such promise inputs and \IND$(x,y) =1$ if and only if $x_i=y_i$.

In \cite[Corollary 3.1]{JW11} they proved the following result:
\begin{theorem}[\cite{JW11}]
Any randomized one-way communication protocol that makes $\delta = 1/4|\U|$
error requires $\Omega(N \log 1/\delta)$ bits of communication.
\end{theorem}
We use this result and prove a lower bound for $\PHam{n,k,k+2}$. 

Let $|\U| = n/k$. Let $f_A: \U \rightarrow \{0,1\}^{3n/k}$, and 
$f_B: (\U \cup \{\perp\}) \rightarrow \{0,1\}^{\frac{3n}{k}}$ be encoding functions defined as
follows:

$f_A: u_i\mapsto A_1A_2\ldots A_{n/k}$, where 
$A_j = \left\{ \begin{array}{cc} 
		110 & \mbox{if~} j=i\\
		000 & \mbox{otherwise}
		\end{array}
	\right .$
 
$f_B: u_i\mapsto B_1B_2\ldots B_{n/k}$, where 
$B_j = \left\{ \begin{array}{cc} 
		011 & \mbox{if~} j=i\\
		000 & \mbox{otherwise}
		\end{array}
	\right .$ 

$f_B(\perp) = 0^{3n/k}$.

\noindent On promise inputs $x,y \in \U^k$, let $F_A(x)$ and $F_B(y)$ be defined as
$f_A(x_1)f_A(x_2)\ldots f_A(x_k)$ and $f_B(y_1)f_B(y_2)\ldots f_B(y_k)$,
respectively.

Under this encoding, it is easy to see that $\PHam{3n,2k,2k+2}(F_A(x), F_B(y))
=1$ if and only if \IND$(x,y)=1$.
Suppose $i+1$ is the first position where $\bot$ appears in $y$.
For each $j<i$ we have $x_j=y_j$ so the Hamming distance of $F_A(x_j)$ and
$F_B(y_j)$ is 2. Also for every position $j>i$ we have $y_j=\bot$ and hence
$F_B(y_j)=0^{3n/k}$, which results in a Hamming distance of 2 between $F_A(x_j)$
and $F_B(x_j)$. So the Hamming distance between $F_A(x)$ and $F_B(y)$ is
$2(k-1)$ plus the 
Hamming distance between $f_A(x_i)$ and $f_B(y_i)$, which is $2$ iff $x_i=y_i$
and 4 iff $x_i\neq y_i$ (since $x_i,y_i\in\U$).

Therefore we get the following lower bound:
\begin{theorem}[Theorem \ref{thm:lb} restated]
Any randomized one-way protocol that makes at most $k/n$ error and computes
$\PHam{3n,2k,2k+2}$, requires $\Omega(k\log(n/k))$ bits. In fact the lower bound
holds for distinguishing between the case $\Delta(x,y) = 2k$ and
$\Delta(x,y)=2k+2$. 
\end{theorem}
By Theorem \ref{thm:ham} this bound is optimal when $n \geq k^2$ (and in fact when $n
\geq k^{1+\epsilon}$, for constant $\epsilon > 0$). 

\end{document}